\documentclass{article}

\usepackage{fullpage}
\usepackage{amsmath,amsthm}
\usepackage{hyperref, url,afterpage}
\usepackage{float}
\usepackage{booktabs, array, multirow}
\usepackage{enumerate}
\usepackage[ruled]{algorithm}
\usepackage[commentColor=black]{algpseudocodex}
\usepackage{enumitem}

\usepackage{thmtools, thm-restate}
\usepackage{tabularray}
\usepackage{float}
\usepackage{subcaption}

\usepackage{tikz}
\usetikzlibrary{arrows.meta}
\usetikzlibrary{calc, arrows}
\usetikzlibrary{shapes.multipart, positioning}

\definecolor{tan} {RGB} {210,180,140}
\definecolor{burntorange} {rgb} {0.746,0.339,0}
\definecolor{violet} {RGB} {168, 50, 117}
\definecolor{green} {RGB} {76,126,22}
\usetikzlibrary{shapes.misc,shapes.geometric,calc}

\newcommand{\lineComment}[1]{\LComment{ #1}}
\newcommand{\rightComment}[1]{\Comment{ #1}}

\newtheorem{theorem}{Theorem}

\newtheorem{lemma}[theorem]{Lemma}
\newtheorem*{lemma*}{Lemma}

\newtheorem*{theorem*}{Theorem}
\newtheorem{definition}{Definition}

\theoremstyle{remark}

\usepackage{xspace}
\newcommand{\congest}{CONGEST\xspace}

\title{Algorithms for Distance Sensitivity Oracles and other Graph Problems on the PRAM\footnote{Authors' affiliation: The University of Texas at Austin, USA; email:  {\tt vigneshm@cs.utexas.edu, vlr@cs.utexas.edu}. This work was supported in part by NSF grant CCF-2008241.}
\author {Vignesh Manoharan  and Vijaya Ramachandran}
}
\date{}

\begin{document}

\maketitle

\begin{abstract}
   % \vignesh{Following paragraph from SPAA proceedings.}
   The distance sensitivity oracle (DSO) problem asks us to preprocess a given graph $G=(V,E)$ in order to answer queries of the form $d(x,y,e)$, which denotes the shortest path distance in $G$ from vertex $x$ to vertex $y$ when edge $e$ is removed. This is an important problem for network communication, and it has been extensively studied in the sequential setting~\cite{DemetrescuTCR08dso,BernsteinK09dso,WeimannY13,Ren22} and recently in the distributed CONGEST model~\cite{ManoharanR25dso}. However, no prior DSO results tailored to the parallel setting were known. 

   % \vignesh{Added last line about other problems.}
   We present the first PRAM algorithms to construct DSOs in directed weighted graphs, that can answer a query in $O(1)$ time with a single processor after preprocessing. We also present the first work-optimal PRAM algorithms for other graph problems that belong to the sequential $\tilde{O}(mn)$ fine-grained complexity class: Replacement Paths, Second Simple Shortest Path, All Pairs Second Simple Shortest Paths and Minimum Weight Cycle. 
\end{abstract}

\section{Introduction}
% \vignesh{Following paragraphs are from  PRAM DSO proceedings}
In a network modeled by a graph $G=(V,E)$, we investigate the problem of computing shortest path distances when an edge $e \in E$ fails. Specifically, we need to answer queries $d(s,t,e)$ for $s,t\in V, e\in E$, where $d(s,t,e)$ is the shortest path distance from $s$ to $t$  in $G$ when edge $e$ is removed. This is a fundamental problem in network communication for routing under an arbitrary edge failure.

Let $|V|=n, |E|=m$. Since explicitly computing all distances would be prohibitively expensive due to output size, we instead consider Distance Sensitivity Oracles (DSOs). In the DSO problem, we first preprocess the graph $G$ to construct an oracle that stores certain information in order to answer queries quickly. Note that without any preprocessing, we can answer a query using a single source shortest path (SSSP) computation. The goal in DSO construction is to obtain a faster query time, ideally constant query time, while keeping the preprocessing cost reasonable.

The DSO problem has received considerable attention in the sequential setting. Starting from~\cite{DemetrescuTCR08dso}, algorithms with different tradeoffs between query time and preprocessing time were obtained, culminating in an efficient algorithm with preprocessing time $\tilde{O}(mn)$, which constructs a distance sensitivity oracle of size $\tilde{O}(n^2)$ that can answer any query in $O(1)$ time~\cite{BernsteinK09dso}. This construction nearly matches the best runtime of APSP. Other algorithms addressed DSO for dense graphs~\cite{Ren22}, and higher query time~\cite{WeimannY13}. Recently, the DSO problem was studied in the distributed CONGEST model~\cite{ManoharanR25dso} where the input graph models the communication network itself. 

Despite the importance of the DSO problem, no non-trivial DSO constructions are known for the parallel setting. 
In this paper, we present three PRAM algorithms for computing DSOs in directed weighted graphs. All our algorithms involve a preprocessing step, after which queries can be answered in $O(1)$ time using a single processor. Our algorithms have a range of tradeoffs between work and parallel time used for preprocessing.

% \vignesh{Following two paragraphs are added, to transition to other problems}
While our main focus in this paper is on the DSO problem, we address other related graph problems as well. If the source and target vertices $s,t\in V$ are fixed, the problem of finding all distances $d(s,t,e)$ is known as the Replacement Paths (RPaths) problem. The closely related Second Simple Shortest Paths problem (2-SiSP) requires us to compute a $s$-$t$ path (for fixed vertices $s,t\in V$) that differs from an initial input $s$-$t$ shortest path $P_{st}$: we need to compute $d_2(s,t) = \min_{e \in P_{st}} d(s,t,e)$. These problems have been studied extensively in the sequential setting~\cite{Yen1971,RodittyZ12,KatohIM82}, and recently in the distributed CONGEST setting~\cite{ManoharanR24rpaths}. Extending to all vertex pairs, we consider the All Pairs Second Simple Shortest Paths problem (2-APSiSP), where we need to compute the 2-SiSP distance $d_2(s,t)$ for all $s,t \in V$ - this problem has been studied in the sequential setting as well~\cite{AgarwalR16sisp}. These graph problems (in directed weighted graphs) belong to the sequential fine-grained $\tilde{O}(mn)$-time complexity class, introduced in~\cite{AgarwalR18finegrained}. This inspires us to consider another graph problem in this class, Minimum Weight Cycle (MWC), which is  the starting point of hardness in this class. 

% \vignesh{Edited}
These related problems have been unexplored in the PRAM setting, and we present work-optimal algorithms for  RPaths, 2-SiSP, 2-APSiSP and MWC. Our algorithms for these problems are more straightforward than our DSO algorithms, as we are able to adapt sequential algorithms leveraging work-optimal parallel algorithms for APSP~\cite{KarczmarzS21} and SSSP~\cite{CaoF23parallel} to obtain work-optimal bounds. 

% In addition to DSO, we consider PRAM algorithms for other graph problems studied in sequential $\tilde{O}(mn)$ fine-grained complexity class~\cite{AgarwalR18finegrained} that DSO (requiring answers to $n$ queries) belongs to. We present algorithms for . Our algorithms for these problems are more straightforward than our DSO algorithms, as we are able to adapt sequential algorithms leveraging work-optimal parallel algorithms for APSP~\cite{KarczmarzS21} and SSSP~\cite{CaoF23parallel} to obtain work-optimal bounds. 

We first state and present our PRAM algorithms for DSO in Sections~\ref{sec:pram:results}-\ref{sec:pram:algb}. Then, we describe and present our results for other problems (RPaths, 2-SiSP, 2-APSiSP, MWC) in Section~\ref{sec:pram:algother}.

\subsection{Preliminaries}
\label{sec:pram:prelim}

% \vignesh{Following paragraph is from thesis preliminaries}
\paragraph{The PRAM model}
\label{sec:prelim:pram}
In the Parallel RAM (PRAM) model~\cite{ShiloachV81}, computation is performed by a number of synchronous processors accessing a shared memory. We specifically use the work-span model to analyze the cost of computation, which avoids concerns such as number of processors and scheduling tasks to processors. The total number of operations performed over all processors is called the \textit{work}. The longest chain of dependencies between these operations is called the \textit{span}. There are different variants of the PRAM model - EREW, CREW, CRCW - depending on whether more than one processor is allowed to read (Exclusive/Concurrent Read) or write (Exclusive/Concurrent Write) to the same location in the shared memory. In our algorithms we use the CREW model, but we can readily translate between these models using known reductions with only polylogarithmic increase in cost. We use the notation $\tilde{O}$ to hide poly-logarithmic factors in our costs.

% \vignesh{Following paragraphs from SPAA proceedings}

Let $G=(V,E)$ be a directed weighted graph, let $|V|=n, |E|=m$. Each edge $(s,t) \in E$ (for $s,t \in V$) has a non-negative integer weight $w(s,t)$. We denote the shortest path distance from $s$ to $t$ by $d(s,t)$. We use $d(s,t,e)$, for $s,t \in V, e \in E$, to denote the shortest path distance from $s$ to $t$ in the graph $G-\{e\}$, i.e., the graph $G$ with edge $e$ removed, also known as the replacement path distance.

% We use the work-depth model to present our algorithms, similar to other PRAM results~\cite{KarczmarzS21,CaoF23parallel}. Our work and time bounds are presented in $\tilde{O}$ format  which hides polylog factors and are valid for EREW, CREW, CRCW models of PRAM due to known standard reductions between these models.

% \vignesh{Definition and context from thesis, modified for just PRAM}
We first define the DSO problem as follows. 

\begin{restatable}{definition}{defdso}
\textit{Distance Sensitivity Oracles (DSO) Problem}: 
Compute $d(x,y,e)$, the $x$-$y$ shortest path distance when $e$ is removed, for any pair of vertices $x,y \in V$ and any $e \in E$. A DSO algorithm is allowed to perform some preprocessing on the graph $G$ to construct on oracle, which then answers queries of the form $d(x,y,e)$. \label{def:dso}
\end{restatable}

In a PRAM algorithm for DSO, we assume that the query can be accessed by any processor. When the query algorithm terminates, the answer to the query is stored in the shared memory.
The parameters of interest in a PRAM algorithm for DSO are as follows:

\begin{enumerate}
    \item Preprocessing work and parallel time: The work and parallel time required during preprocessing to construct the DSO.
    \item Query work and parallel time: The work and parallel time to answer a distance query.
    \item Oracle Size: Total number of words stored for the oracle after preprocessing.
\end{enumerate}

In this paper, we primarily focus on DSOs with a query cost of $O(1)$ work on a single processor.

We will also use the notion of hop-limited DSO, defined as follows:

\begin{definition}
    Let $d_h(x,y,e)$, for any $1 \le h \le n$, denote the minimum weight path among all $x$-$y$ paths of at most $h$ edges not containing edge $e$. In an $h$-hop limited DSO , also called $h$-hop DSO, we need to preprocess an input graph $G=(V,E)$ to answer queries $d_h(x,y,e)$ for any $x,y \in V, e \in E$. Note that an $n$-hop DSO is simply a complete DSO.
    \label{def:pram:hdso}
\end{definition}

\subsection{Roadmap}
\label{sec:pram:roadmap}
This paper is organized as follows. We begin in Section~\ref{sec:pram:results} by stating our PRAM algorithm results for DSO.
We describe the sequential DSO framework of~\cite{BernsteinK09dso} in Section~\ref{sec:pram:frame}, which serves as a foundation for our constructions. In Section~\ref{sec:pram:alga}, we present \textsc{DSO-A}, our first PRAM algorithm, which is a work-efficient parallel implementation of this framework. Next, in Section~\ref{sec:pram:algb}, we develop our faster algorithms, \textsc{DSO-B} and \textsc{DSO-C}. Finally in Section~\ref{sec:pram:algother}, we present work-optimal PRAM algorithms for a suite of related graph problems: RPaths, 2-SiSP, 2-APSiSP, MWC and ANSC.

\subsection{Prior Work}

% \vignesh{Added prior work based on thesis.}
The DSO problem has been studied extensively in the sequential setting~(e.g. \cite{DemetrescuTCR08dso, BernsteinK08dso, BernsteinK09dso, WeimannY13, Ren22, DeyG24}). No non-trivial algorithms for DSO were known in the PRAM setting. The other graph problems studied in this paper (such as RPaths, MWC, and 2-APSiSP) also have not been previously investigated in the PRAM model.

PRAM algorithms for shortest path problems in graphs have been studied extensively: The first algorithms for SSSP tailored to PRAM were given in~\cite{KleinS97, Spencer97}, with the current best work-optimal algorithm for directed exact SSSP taking $\tilde{O}(m)$ work and $\tilde{O}(n^{1/2+o(1)})$ parallel time~\cite{CaoF23parallel}. Approximate SSSP has also been studied, with improved $\tilde{O}(1)$ parallel time for undirected graphs~\cite{RozhonGHZL22undir}.

APSP can be solved in the PRAM in $\tilde{O}(n^3)$ work and $\tilde{O}(1)$ parallel time~\cite{HanPR97}. For sparser graphs, a deterministic algorithm for directed weighted APSP was given in~\cite{KarczmarzS21} that takes $\tilde{O}(mn+(n/d)^3)$ work and $\tilde{O}(d)$ parallel time, for a parameter $d \in [1,n]$. This improves on a randomized parallel transitive closure trade-off of~\cite{UllmanY91}.

% \vignesh{All following sections are from thesis, which built on SPAA proceedings.}

\section{Summary of DSO Results}
\label{sec:pram:results}

We present PRAM algorithms for DSO in directed weighted graphs with integer weights, our results are tabulated in Table~\ref{tab:pram:results}. 
Integer weights are needed in our algorithms only for computing SSSP where we utilize the PRAM algorithms of~\cite{CaoF23parallel,KleinS97}.  Otherwise, our results hold for real weighted graphs. All of our preprocessing algorithms are randomized, and are correct w.h.p. in $n$. 
In all of our algorithms, we construct an $\tilde{O}(n^2)$-sized oracle that can answer any query in $O(1)$ time on a single processor. 

\begin{table}[t!]
    \caption[PRAM results for directed weighted DSO]{ PRAM results for directed weighted DSO. All our algorithms construct an $\tilde{O}(n^2)$-sized oracle that can answer any query in $O(1)$ time on a single processor.}
    \label{tab:pram:results}
    \begin{center}
        \begin{tblr}{hlines, vlines, columns={c}, column{1}={3.4cm}, column{2}={3cm}, column{3}={3cm}, column{4}={1.4cm}, column{5}={1.4cm}}
            \textbf{Result} & \textbf{Preprocessing work} & \textbf{Preprocessing span/time} & \textbf{Query work} & \textbf{Oracle size} \\ 
            {\textit{Algorithm DSO-A}} & {{$\tilde{O}(mn)$}}  & {{$\tilde{O}(n^{1/2+o(1)})$}} & {$O(1)$} & {$\tilde{O}(n^2)$} \\
            {\textit{Algorithm DSO-B}} & {$\tilde{O}(n^3)$}  & {$\tilde{O}(1)$} & {$O(1)$} & {$\tilde{O}(n^2)$} \\
            {\textit{Algorithm DSO-C} \\ (for any $1 \le h \le n$)} & {$\tilde{O}(mnh+\frac{n^3}{h})$}  & {$\tilde{O}(h)$} & {$O(1)$} & {$\tilde{O}(n^2)$} \\
        \end{tblr}
    \end{center} 
\end{table}

Our first algorithm \textsc{DSO-A} with bounds shown in Theorem~\ref{thm:pram:alga} implements the sequential DSO of~\cite{BernsteinK09dso} on the PRAM. This algorithm is work-efficient for a constant query cost DSO, matching the preprocessing work, oracle size, and query time of the sequential algorithm in~\cite{BernsteinK09dso}. This is described in Section~\ref{sec:pram:alga}.

\begin{restatable}{theorem}{thmpramalga}
We can construct a DSO for directed weighted graphs on the PRAM with preprocessing cost of $\tilde{O}(mn)$ work and $\tilde{O}(n^{1/2+o(1)})$ parallel time. The constructed oracle has size $\tilde{O}(n^2)$, and can answer any query in $O(1)$ work. 
    \label{thm:pram:alga}
\end{restatable}

Our next two algorithms use the notion of hop-limited DSOs, see Definition~\ref{def:pram:hdso}. We use a PRAM implementation of a sequential technique from~\cite{Ren22} to construct a $\frac{3}{2}h$-hop limited DSO given a $h$-hop limited DSO ($1 \le h \le n$) at the cost of increased preprocessing.
This extension procedure, described in Section~\ref{sec:pram:extend}, consists of two steps: constructing a $\frac{3}{2}h$-hop limited DSO with high query time using vertex sampling (in Section~\ref{sec:pram:highquery}), and then reducing the query time to $O(1)$ using the framework of~\cite{BernsteinK09dso} and the PRAM APSP algorithm of~\cite{KarczmarzS21} (in Section~\ref{sec:pram:lowquery}).

We present a fast algorithm \textsc{DSO-B} with bounds shown in Theorem~\ref{thm:pram:algb} that is work efficient for dense graphs ($m=\Theta(n^2)$). The work also matches the current best work of PRAM APSP in $\tilde{O}(1)$ parallel time. This result uses a straightforward 2-hop DSO construction as a base case together with repeated application of the parallel DSO extension procedure. It is described in Section~\ref{sec:pram:algb:dso1}. 

\begin{restatable}{theorem}{thmpramalgb}
We can construct a DSO for directed weighted graphs on the PRAM with preprocessing cost of $\tilde{O}(n^3)$ work and $\tilde{O}(1)$ parallel time. The constructed oracle has size $\tilde{O}(n^2)$, and can answer any query in $O(1)$ work. 
    \label{thm:pram:algb}
\end{restatable}

Theorem \ref{thm:pram:algc} gives a tradeoff between work and parallel time in algorithm \textsc{DSO-C}. 
In Theorem~\ref{thm:pram:algc}, we can choose $h = \tilde{o}\left( \min \left( \frac{n^2}{m}, \sqrt{n} \right) \right)$ to obtain an algorithm that achieves sub-$\sqrt{n}$ parallel time with sub-$n^3$ work as long as $m = \tilde{o}(n^2)$, i.e., the graph is not fully dense. Note that this tradeoff is only meaningful for $h \le \tilde{O}(n^{1/2+o(1)})$, as Theorem~\ref{thm:pram:alga} provides a superior bound for larger $h$.
% Theorem~\ref{thm:pram:algc} is only meaningful for $h \le \tilde{O}(n^{1/2+o(1)})$ due to Theorem~\ref{thm:pram:alga}. 

To achieve Theorem \ref{thm:pram:algc}, we use a hop limited approach similar to Theorem~\ref{thm:pram:algb}, but we use a more sophisticated base case construction that directly constructs a $h$-hop limited DSO. This is done with a graph sampling technique, using ideas from a sequential DSO algorithm~\cite{WeimannY13}, and a parallel hop-limited shortest paths algorithm~\cite{CaoF23parallel}. After the base case construction, we repeatedly apply the parallel DSO extension procedure to obtain the result, proven in Section~\ref{sec:pram:algb:dso2}.

\begin{restatable}{theorem}{thmpramalgc}
We can construct a DSO for directed weighted graphs on the PRAM with preprocessing cost of $\tilde{O}(mnh + (n^3/h))$ work and $\tilde{O}(h)$ parallel time, for any $1 \le h \le n$. The constructed oracle has size $\tilde{O}(n^2)$, and can answer any query in $O(1)$ work. 
    \label{thm:pram:algc}
\end{restatable}

\section{Framework for Distance Sensitivity Oracle}
\label{sec:pram:frame}

Our PRAM algorithms for DSO use techniques from sequential DSO algorithms~\cite{DemetrescuTCR08dso,BernsteinK09dso,Ren22,WeimannY13}. In this section, we describe a framework for DSO construction due to Bernstein and Karger~\cite{BernsteinK09dso}. In Sections~\ref{sec:pram:alga} and~\ref{sec:pram:algb}, we will use different implementations of this framework along with other techniques in our PRAM algorithms. We present our results for edge removal (this can be modified for vertex removal as in~\cite{BernsteinK09dso}).
Recall that the sequential DSO construction of~\cite{BernsteinK09dso} has $\tilde{O}(mn)$ preprocessing time, $\tilde{O}(n^2)$ space requirement, and $O(1)$ query time. 

% \vignesh{Removed some comparisons to DSO algorithm.}
% Note that we used elements of a different sequential DSO algorithm from~\cite{BernsteinK08dso} in one of our \congest DSO algorithms (Algorithm~\ref{alg:dso:preproc} in Section~\ref{sec:dso:algpre}): This sequential algorithm is slower than the result in~\cite{BernsteinK09dso}, but is more amenable to the distributed setting. Due to the difference in challenges in the PRAM model, we are able to implement~\cite{BernsteinK09dso}.

The DSO framework of~\cite{BernsteinK09dso} is described in Algorithm~\ref{alg:pram:frame}. We described the essential aspects of this algorithm now.

\begin{algorithm}[t!]
    \caption{Framework from~\cite{BernsteinK09dso} for DSO with $O(1)$ query.}
    \begin{algorithmic}[1]
        \Require Directed weighted graph $G=(V,E)$.
        \State \textbf{Preprocessing Algorithm} :
        \State Every vertex is randomly assigned a priority $k$ with probability $\Theta(\frac{1}{2^k})$, for $k=1,2,\dots \log n$. \label{alg:pram:frame:samp}
        \State For each $x,y \in V$, $i \le \log n$, compute $CR[x,y,i], CL[x,y,i], BCP[x,y]$. \label{alg:pram:frame:apsp}
        \State For each $x \in V$ and edge $e$ covered by $x$, for each $y \in V$, compute $D_i[x,y,e]$ ($i = $ priority of $x$). \label{alg:pram:frame:exclude}
        \State For each $x,y \in V$, $i \le \log n$, compute bottleneck vertex $BV[x,y,i]$ using binary search and an RMQ data structure on the computed $D_i[x,y,e]$ values. \label{alg:pram:frame:binary}
        \State For each $x,y \in V$, $i \le \log n$, compute distance $DBV[x,y,i]$ when vertex $BV[x,y,i]$ is removed. \label{alg:pram:frame:dbv}
        \State \textbf{Query Algorithm}, given query $d(x,y,e)$ :
        \State Let edge $e=(u,v)$. Let $i=BCP[x,u]$, $j=BCP[v,y]$.
        \State Compute $c_x = CL[x,y,i]$ and $c_y=CR[x,y,j]$. \rightComment{Edge $e$ is covered by $c_x$ and $c_y$}
        \State Output $d(x,y,e) = \min\{ d(x,c_x)+D_i[c_x,y,e], D_i[x,c_y,v]+d(c_y,y), DBV[x,y,i]\}$
    \end{algorithmic}
    \label{alg:pram:frame}
\end{algorithm}

Each vertex in the input graph $G=(V,E)$ is assigned a \textit{priority} $k$, where $k \in \{1,2, \dots \log n\}$, and a vertex with priority $k$ is called a $k$-\textit{center}. In each $x$-$y$ shortest path, for $x,y \in V$, the algorithm identifies a sequence of centers of strictly increasing priority up to the maximum, followed by strictly decreasing priority to $y$. A center $x$ \textit{covers} an edge $e$, if $e$ is on its outgoing (or incoming) shortest path tree, and there is no center of higher priority on the shortest path from $x$ to $e$. Thus, the sequence of centers partitions the path into intervals that are covered by their endpoint centers. By randomly sampling $k$-centers with probability $\Theta(1/2^k)$, a center of priority $k$ only has to cover edges within $\tilde{O}(2^k)$ hops from it.

The algorithm precomputes distances $d(x,y,e)$ for all $y \in V$ and edges $e \in E$ that are covered by center $x$. Furthermore, the algorithm uses the notion of a \textit{bottleneck edge} for an interval, which is the edge in the interval that maximizes the replacement path distance when it is removed. The following data is computed during preprocessing, where $x,y \in V, e \in E$ and $i \in \{1,2, \dots 2\log n\}$ :

\begin{itemize}
    \item $CR[x,y,i]$: first center of priority-$\ge i$ on $x$-$y$ shortest path.
    \item $CL[x,y,i]$: first center of priority-$\ge i$ on reversed $y$-$x$ shortest path (edge directions flipped). 
    \item $BCP[x,y]$: biggest center priority on $x$-$y$ shortest path.
    \item $D_i[x,y,e]$: distance $d(x,y,e)$, where $x$ is a center of priority $i$ that covers $e$.
    \item $BV[x,y,i]$: bottleneck edge of interval $i$ on $x$-$y$ shortest path.
    \item $DBV[x,y,i]$: distance $d(x,y,BV[x,y,i])$, for each interval $i$.
\end{itemize}

In the PRAM model, we can store these values in the shared memory and utilize them to answer queries. The preprocessed data takes up a total of $\tilde{O}(n^2)$ space. After preprocessing, a query can be answered by looking up $O(1)$ values, as shown in Algorithm~\ref{alg:pram:frame}. Note the distinction from the distributed DSO algorithms in~\cite{ManoharanR25dso} had to perform additional work, as these values may be computed at different nodes and then need to be communicated.

In the next two sections, we present PRAM algorithms that use this framework in different ways. All our algorithms finally compute the same preprocessed data as~\cite{BernsteinK09dso}, so the oracle can be stored in $\tilde{O}(n^2)$ space and can answer a query in $O(1)$ work on a single processor. 
In Section~\ref{sec:pram:alga} we show how to directly implement Algorithm~\ref{alg:pram:frame} efficiently on the PRAM, to obtain a work-efficient preprocessing algorithm that runs in $\tilde{O}(n^{1/2+o(1)})$ parallel time. In Section~\ref{sec:pram:algb}, we use this framework in addition to hop-limited DSO construction techniques to obtain faster preprocessing algorithms.

\section{Work-Efficient PRAM DSO}
\label{sec:pram:alga}

In this section, we prove Theorem~\ref{thm:pram:alga} by presenting \textsc{DSO-A} in Algorithm~\ref{alg:pram:alga}. \textsc{DSO-A} implements the framework of Algorithm~\ref{alg:pram:frame} in PRAM, and preprocesses the input graph in $\tilde{O}(mn)$ work and $\tilde{O}(n^{1/2+o(1)})$ parallel time. The constructed oracle has $\tilde{O}(n^2)$ size and can  answer any query in $O(1)$ work. This algorithm uses excluded shortest paths computations, where given a graph $G=(V,E)$, set of edges $P \subseteq E$ and a source $x \in V$, we need to compute distances $d(x,y,e)$ for all $y\in V, e \in P$. If the set $P$ is \textit{independent}, i.e., the edges are in the outgoing shortest path tree of $x$ and the subtrees rooted at any pair of edges in $P$ are disjoint, then we can compute all excluded shortest paths distances using one SSSP computation~\cite{DemetrescuTCR08dso}.

\begin{algorithm}[t!]
    \caption{\textsc{DSO-A} implements Algorithm~\ref{alg:pram:frame} in PRAM with preprocessing cost $\tilde{O}(mn)$ work and $\tilde{O}(n^{1/2+o(1)})$ parallel time, $O(1)$ work query and $O(n^2)$ oracle size.}
    \begin{algorithmic}[1]
        \Require Directed weighted graph $G=(V,E)$.
        \State \textbf{Preprocessing Algorithm} :
        \State Vertex $v$ is randomly assigned a \textit{priority} $p_v$ from $\{1,2, \dots \lceil \log n \rceil\}$, with a probability distribution choosing priority $j$ with probability $c/2^j$, for appropriate constant $c$. \label{alg:pram:alga:samp}
        \For{vertex $x \in V$}
        \State Compute SSSP from $x$ in $G$. Along with distances $d(x,y)$ for $y \in V$, also compute $CR[x,y,i], CL[x,y,i], BCP[x,y]$ for $y \in V, 1 \le i \le \log n$. \label{alg:pram:alga:apsp}
        \State Compute excluded shortest paths distances from $x$ for each edge $e$ that covers, i.e., $D_i[x,y,e]$ values. The set of edges at same depth in the shortest path tree $T_x$ can be excluded at the same time, as they are {independent}. \label{alg:pram:alga:exclude} \rightComment{W.h.p. in $n$, $\tilde{O}(2^{p_x})$ exclude computations are performed for source $x$. }
        \EndFor
        \For{vertices $x,y \in V$, priority $1 \le i \le \log n$}
        \State Compute bottleneck vertex $BV[x,y,i]$ using binary search and a parallel RMQ data structure~\cite{BerkmanV93} on the computed $D_i[x,y,e]$ values. \label{alg:pram:alga:binary}
        % \State Compute distance $DBV[x,y,i]$ when vertex $BV[x,y,i]$ is removed. \label{alg:pram:frame:dbv}
        \EndFor
        \For{vertex $x \in V$}
        \State Construct graph $G_x$ with vertices $(x,y,i), \forall y\in V,i \le \log n$ and a source vertex $s$, and $\tilde{O}(m+n)$ edges, following the construction in~\cite{BernsteinK09dso}. Compute SSSP on this graph to compute $DBV[x,y,i]$. \label{alg:pram:alga:dbv}
        \EndFor
        % \State For each $x,y \in V$, $i \le \log n$, compute $CR[x,y,i], CL[x,y,i], BCP[x,y]$, . \label{alg:pram:frame:apsp}
        % \State For each $x \in V$ and edge $e$ covered by $x$, for each $y \in V$, compute $D_i[x,y,e]$ ($i = $ priority of $x$). \label{alg:pram:frame:exclude}

        \State \textbf{Query Algorithm}, given query $d(x,y,e)$ :
        \State Let edge $e=(u,v)$. Let $i=BCP[x,u]$, $j=BCP[v,y]$.
        \State Compute $c_x = CL[x,y,i]$ and $c_y=CR[x,y,j]$.
        \State Output $d(x,y,e) = \min\{ d(x,c_x)+D_i[c_x,y,e], D_i[x,c_y,v]+d(c_y,y), DBV[x,y,i]\}$ \rightComment{$D_i$ values are computed in line~\ref{alg:pram:alga:exclude} and $DBV$ is computed in line~\ref{alg:pram:alga:dbv}}
    \end{algorithmic}
    \label{alg:pram:alga}
\end{algorithm}

\thmpramalga*

\begin{proof}[Proof of Theorem~\ref{thm:pram:alga}]
    We show how Algorithm~\ref{alg:pram:alga} implements Algorithm~\ref{alg:pram:frame} on the PRAM. 
    
    In line~\ref{alg:pram:alga:apsp} of Algorithm~\ref{alg:pram:alga}, we compute the center information $CR[x,y,i]$ and $BCP[x,y]$. Along with SSSP, with an additional $O(\log n)$ factor cost, we keep track of the closest vertex on each shortest path with given priority $i$. We repeat this on the reversed graph for $CL[x,y,i]$. 

    In line~\ref{alg:pram:alga:exclude}, note that each center $x$ of priority $k$ covers edges within $\tilde{O}(2^k)$ depth in its outgoing shortest path tree w.h.p. in $n$. The set of edges at a given depth $j$ in the outgoing shortest path tree of $x$ is independent, so we use one excluded shortest paths computation for these edges. For depths $1 \le j \le \tilde{O}(2^k)$, we use $\tilde{O}(2^k)$ exclude computations, and each exclude involves one SSSP computation. Since we have $\tilde{O}(n/2^k)$ vertices of priority $k$, and $O(\log n)$ priorities, we get a total of $\tilde{O}(n)$ SSSP computations over all $x \in V$ w.h.p. in $n$. Finding bottleneck vertices in line~\ref{alg:pram:alga:binary} uses a range minimum query (RMQ) data structure, which is implemented on the PRAM with linear work and $O(\log^* n)$ parallel time for preprocessing and $O(1)$ work per query~\cite{BerkmanV93}. 
    
    In line~\ref{alg:pram:alga:dbv}, we perform $\tilde{O}(n)$ SSSP computations. The algorithm of~\cite{BernsteinK09dso} constructs a graph with vertices $(x,y,i)$ for each interval $i$ on $x$-$y$ shortest path for $x,y \in V, i \le \log n$, and an additional vertex $s$. The edges are between $(x,y,i)$ to $(x,y',j)$ with weight $w(y',y)$ for each edge $(y',y) \in E$, and additional edges from $s$ to $(x,y,i)$. Then, they directly compute SSSP on this graph. We instead construct $n$ graphs $G_x$, one for each $x \in V$, with vertex $s$ and vertices $(x,y,i), \forall y\in V,i \le \log n$. Since there are no edges between different $G_x, G_y$, we can compute SSSP separately in each $G_x$, which has $\tilde{O}(n)$ vertices and $\tilde{O}(m+n)$ edges. This computes $DBV[x,y,i]$ as the shortest distance from $s$ to vertex $(x,y,i)$ in $G_x$.

    In total, we perform $\tilde{O}(n)$ SSSP computations, which takes a total of $\tilde{O}(mn)$ work and $\tilde{O}(n^{1/2+o(1)})$ time using the PRAM algorithm of~\cite{CaoF23parallel}. Line~\ref{alg:pram:alga:binary} takes $\tilde{O}(n^2)$ work and $\tilde{O}(1)$ time, and other lines take $O(n)$ work.
    The correctness of our algorithm readily follows from the correctness of~\cite{BernsteinK09dso}, and the oracle has the same $\tilde{O}(n^2)$ size as the construction of~\cite{BernsteinK09dso}. The query algorithm is the same as in Algorithm~\ref{alg:pram:frame}, which takes $O(1)$ work per query.
\end{proof}

\section{Faster PRAM DSO for Dense Graphs}
\label{sec:pram:algb}

In this section, we present PRAM algorithms that can beat the $\tilde{O}(n^{1/2+o(1)})$ parallel time of the preprocessing algorithm presented in Section~\ref{sec:pram:alga}. 
Our method constructs a series of hop limited DSOs (see Definition~\ref{def:pram:hdso}) in order to ultimately obtain a complete DSO.  We first describe the hop limited DSO constructions, and then show how we use it in our algorithms.

\subsection{Hop-limited DSO}
\label{sec:pram:extend}

% Let $d_h(x,y,e)$, for any $1 \le h \le n$, denote the minimum weight path among all $x$-$y$ paths of at most $h$ edges not containing edge $e$. In an $h$-hop limited DSO we needed to preprocess an input graph $G=(V,E)$ to answer queries $d_h(x,y,e)$ for any $x,y \in V, e \in E$. Note that an $n$-DSO is simply a complete DSO.
In this section, we describe a two-step DSO extension procedure to construct $\frac{3}{2}h$-hop limited DSO from an $h$-hop limited DSO, adapting a sequential technique of Ren~\cite{Ren22}.

\begin{lemma}
    Let $1 \le h \le n$.
    Given an $h$-hop DSO with preprocessing cost $P_w$ work and $P_t$ parallel time ($P_t \le n$) which answers a query in $\tilde{O}(1)$ work, we can construct a $\frac{3}{2}h$-hop DSO with preprocessing cost $P_w+\tilde{O}(mn+(n^3/h)+(n/P_t)^3)$ work and $P_t + \tilde{O}(1)$ parallel time, which can answer any query in $O(1)$ work.
    \label{lem:pram:extend}
\end{lemma}
\begin{proof}
    The two-step procedure is as follows:
    The first step, proven in Lemma~\ref{lem:pram:highquery} of Section~\ref{sec:pram:highquery}, constructs a $\frac{3}{2}h$-hop DSO with preprocessing cost of $P_w+{O}(n)$ and $P_t+{O}(1)$ parallel time and cost per query of $\tilde{O}(n/h)$ work and $\tilde{O}(1)$ parallel time.

    The second step, proven in Lemma~\ref{lem:pram:lowquery} of Section~\ref{sec:pram:lowquery}, reduces the query cost to $O(1)$ work. We construct a $\frac{3}{2}h$-hop DSO with $O(1)$ work per query and preprocessing cost of $P_w+\tilde{O}(mn+(n^3/h)+(n/P_t)^3)$ work and $P_t + \tilde{O}(1)$ time.
\end{proof}

\subsubsection{Extended hop DSO with high query time}
\label{sec:pram:highquery}
Assume that we are given a $h$-hop limited DSO with $O(1)$ query time. The following observation is from~\cite{Ren22}: Let $s$ be a vertex on a $\frac{3}{2}h$-hop replacement path from $u$ to $v$ such that $s$ is at most $h$-hops from both $u$ and $v$ along this path. Then, $d_{\frac{3}{2}h}(u,v,e) = d_h(u,s,e)+d_h(s,v,e)$, as the two subpaths $u$-$s$ and $s$-$v$ have distances $d_h(u,s,e)$ and $d_h(s,v,e)$ respectively. We use this result in the following lemma.

\begin{lemma}
    Given an $h$-hop DSO with $P_w$ work and $P_t$ parallel time for preprocessing and $\tilde{O}(1)$ work for query, we can construct a $\frac{3}{2}h$-hop DSO with preprocessing cost $P_w+{O}(n)$ work and $P_t + {O}(1)$ time, which answers a query in $\tilde{O}(n/h)$ work and $\tilde{O}(1)$ time. \label{lem:pram:highquery}
\end{lemma}
\begin{proof}
    Sample each vertex into a set $S$ with probability $\Theta(\log n/h)$ so that any path of length $\frac{h}{2}$ is hit by a sampled vertex w.h.p in $n$. Note that $|S|=\tilde{O}(n/h)$ w.h.p. in $n$. To answer a query, we compute $ d_{\frac{3}{2}h}(u,v,e) = \min_{s \in S} d_h(u,s,e)+d_h(s,v,e)$. This is correct w.h.p. in $n$: The path with distance $d_h(u,s,e)+d_h(s,v,e)$ is a valid $u$-$v$ path not containing $e$ for any $s \in V$. Consider the subpath of the $\frac{3}{2}h$-hop replacement path containing vertices within $h$ hops of both $u$ and $v$. This segment has length of at least $\frac{h}{2}$ and therefore has a sampled vertex $s \in S$ w.h.p. in $n$. This $s$ gives us the correct distance $d_{\frac{3}{2}h}(u,v,e)$ by the above observation.

    The preprocessing algorithm of the $\frac{3}{2}h$-hop DSO is to first sample and store $S$, which takes $O(n)$ work and $O(1)$ time. Then, we run the preprocessing algorithm of the input $h$-hop DSO.
    A query computes the minimum over $O(|S|)=\tilde{O}(n/h)$ $h$-hop distances, which are computed using queries to the $h$-hop DSO. This takes $\tilde{O}(n/h)$ work and $\tilde{O}(1)$ time per $\frac{3}{2}h$-hop query.
\end{proof}

\subsubsection{Reducing to $O(1)$ query time}
\label{sec:pram:lowquery}

Now, we convert a given $h$-hop DSO with high query time into an $h$-hop DSO with $O(1)$ query time using additional preprocessing by implementing Algorithm~\ref{alg:pram:frame}. 

\begin{lemma}
    Given an $h$-hop DSO with $P_w$ work and $P_t$ parallel time for preprocessing and $q$ work and $\tilde{O}(1)$ time for query, we can construct a $\frac{3}{2}h$-hop DSO with $P_w+\tilde{O}(mn+n^2\cdot q+(n/P_t)^3)$ work and $P_t + \tilde{O}(1)$ time for preprocessing and $O(1)$ work for query. \label{lem:pram:lowquery}
\end{lemma}
\begin{proof}
We show how we implement each preprocessing step of Algorithm~\ref{alg:pram:frame}, after which we can use the query algorithm described there. Line~\ref{alg:pram:frame:binary} is implemented in $\tilde{O}(n^2)$ work and $\tilde{O}(1)$ time as in Section~\ref{sec:pram:alga}. Lines~\ref{alg:pram:frame:apsp},\ref{alg:pram:frame:exclude},\ref{alg:pram:frame:dbv} are implemented differently, as follows.

Line~\ref{alg:pram:frame:apsp}: To compute $CR, CL$ and $BCP$ values for each interval on the shortest path between each vertex pair $x,y \in V$, for a fixed $i$, we do a modified APSP. To compute $CR[x,y,i]$, we track the first center of priority $\ge i$ in addition to $d(x,y)$. When combining distances during APSP, we also update the $CR$ value. This modification is readily done to the PRAM APSP algorithm in~\cite{KarczmarzS21}. The values $CL, BCP$ are computed in a similar way. For $1 \le i \le \log n$, we perform $O(\log n)$ APSP computations which takes $\tilde{O}(mn+(n/P_t)^3)$ work and parallel time $P_t$ using the algorithm of~\cite{KarczmarzS21}.

Lines~\ref{alg:pram:frame:exclude},\ref{alg:pram:frame:dbv}: These lines require computing $h$-hop replacement distances, for which we use queries to the input $h$-hop DSO with $q$ work per query. We run the preprocessing algorithm of the given DSO, in parallel to line~\ref{alg:pram:frame:apsp}. Then,
 in line~\ref{alg:pram:frame:exclude}, we compute distance $d(x,y,e)$ for each edge $e$ on the $x$-$y$ shortest path that is covered by center $x$. The number of such distances for a fixed vertex $y$ is $\le 2n$ as each edge is covered by the endpoint centers of the interval containing it. This is a total of $O(n^2)$ distances for all $y \in V$. Line~\ref{alg:pram:frame:dbv} thus computes $\tilde{O}(n^2)$ distances. So, we have an additional preprocessing cost of $\tilde{O}(n^2q)$ work and $\tilde{O}(1)$ time.
\end{proof}

\subsection{Faster DSO constructions for Dense Graphs}
\label{sec:pram:algb:dso}
We now present our DSO preprocessing algorithms that improve on $\tilde{O}(\sqrt{n})$ parallel time.

\subsubsection{Preprocessing algorithm for DSO in $\tilde{O}(n^3)$ work and $\tilde{O}(1)$ time}
\label{sec:pram:algb:dso1}
In this section, we present algorithm \textsc{DSO-B}, visualized in Figure~\ref{fig:pram:algb}, that takes $\tilde{O}(n^3)$ work and $\tilde{O}(1)$ parallel time to construct an oracle of size $\tilde{O}(n^2)$, that can answer a query in $O(1)$ time. We first construct a 2-hop limited DSO with high query time and apply Lemma~\ref{lem:pram:lowquery} to obtain a 2-hop DSO with $O(1)$ query time. Then, we repeatedly apply Lemma~\ref{lem:pram:extend} for $O(\log n)$ steps until we obtain a $n$-hop DSO -- which is just a complete DSO.

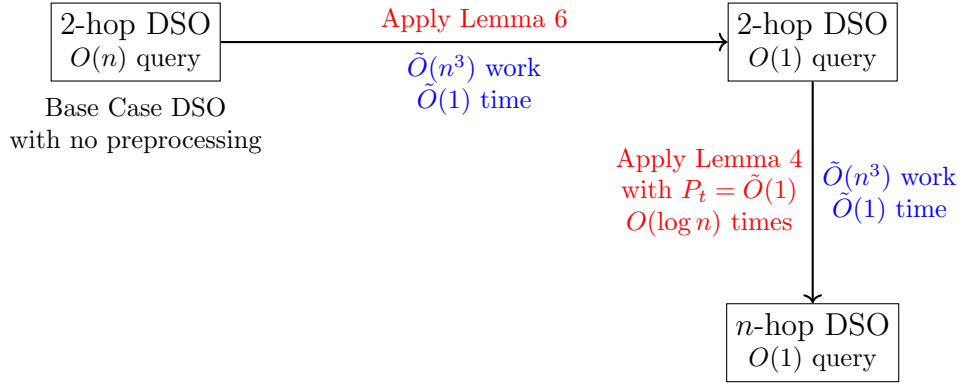
\begin{figure}
    \centering
    \begin{tikzpicture}
        \tikzset{shp/.style={color=black,draw}}
        \tikzset{every text node part/.style={align=center}}
        \tikzstyle{vertex}=[circle, draw=black, fill=black!50, minimum size=4pt, line width=0.3pt, inner sep=0pt]
        \tikzstyle{gedge}=[line width=0.7pt]

        % \node[vertex] (v1) at (-1,-2.5) {};
        % \node  [below=0.05cm of v1] {$x$};
        % \node[vertex, fill=red!50] (v2) at (0,-2) {};
        % \node[vertex] (v3) at (1,-2.5) {};
        % \node  [below=0.05cm of v3] {$y$};
        % \path[draw,thick,->,gedge] (v1) edge (v2);
        % \path[draw,thick,->,gedge] (v2) edge (v3);
        
        \node (h1) [shape=rectangle,shp] at (0,0) {{\large $2$-hop DSO} \\ $O(n)$ query};
        \node (h1txt) [below=0.1cm of h1] {Base Case DSO \\ with no preprocessing };

        \node (h2) [shape=rectangle,shp] at (9,0) {{\large $2$-hop DSO} \\ $O(1)$ query};
        \path[draw, thick, ->] (h1) edge node [above] {\textcolor{red}{Apply Lemma~\ref{lem:pram:lowquery}}} node [below] {\textcolor{blue}{$\tilde{O}(n^3)$ work} \\ \textcolor{blue}{$\tilde{O}(1)$ time}} (h2) ;

        \node (h3) [shape=rectangle,shp] at (9,-4) {{\large $n$-hop DSO} \\ $O(1)$ query};
        \path[draw, thick, ->] (h2) edge node [left] {\textcolor{red}{Apply Lemma~\ref{lem:pram:extend}} \\ \textcolor{red}{with $P_t=\tilde{O}(1)$} \\ \textcolor{red}{ $O(\log n)$ times}} node [right] {\textcolor{blue}{$\tilde{O}(n^3)$ work} \\ \textcolor{blue}{$\tilde{O}(1)$ time}} (h3) ;

        % \node (h4) [shape=rectangle,shp] at (9,-8) ;
        % \path[draw, thick, ->, dashed] (h3) edge node [left] {\textcolor{red}{Apply Lemma~\ref{lem:pram:extend}} \\ \textcolor{red}{ $O(\log n)$ times}} node [right] {\textcolor{blue}{$\tilde{O}(n^3)$ work} \\ \textcolor{blue}{$\tilde{O}(1)$ time}} (h4) ;

    \end{tikzpicture}
    \caption{Visual representation of preprocessing algorithm \textit{DSO-B}.}
    \label{fig:pram:algb}
\end{figure}

\textit{Base Case:}
Our base DSO is a 2-hop limited DSO, with a query cost of $O(n)$ work, and $\tilde{O}(1)$ parallel time. This needs no preprocessing, as $d_2(x,y,e) = \min_{s \in V} w(x,s)+w(s,y)$, where the minimum is over $s \in V$ such that $(x,s), (s,y) \in E$ and $e \ne (x,s), (s,y)$. We now apply Lemma~\ref{lem:pram:lowquery} to reduce this query time, with parameters $P_t=\tilde{O}(1), q=n$, to obtain a 2-hop DSO with preprocessing cost of $\tilde{O}(n^3)$ work and $\tilde{O}(1)$ parallel time and $O(1)$ work query.

\textit{Obtaining complete DSO:}
We apply Lemma~\ref{lem:pram:extend} with $P_w=\tilde{O}(n^3)$ and $P_t=\tilde{O}(1)$. We repeat this procedure $O(\log n)$ times until we obtain a $n$-hop DSO with $O(1)$ query. Over all $O(\log n)$ steps, we get a total of $\tilde{O}(n^3)$ work and $\tilde{O}(1)$ parallel time. 

Note that the final query is just looking up $O(1)$ values precomputed by the $n$-hop DSO. Thus, we only maintain preprocessed data for the $n$-hop DSO, which has size $\tilde{O}(n^2)$, and the intermediate DSOs can be discarded.

\subsubsection{DSO with improved work-time tradeoff}
\label{sec:pram:algb:dso2}
In this section, we present algorithm \textsc{DSO-C}, visualized in Figure~\ref{fig:pram:algc}, that performs $\tilde{O}(mnh + (n^3/h))$ work and $\tilde{O}(h)$ parallel time for preprocessing, for any $2 \le h \le n$, to construct an oracle of size $\tilde{O}(n^2)$ that answers a query in $O(1)$ work.
Our preprocessing algorithm follows a scheme similar to \textsc{DSO-B} in the previous section, but uses a different base case.

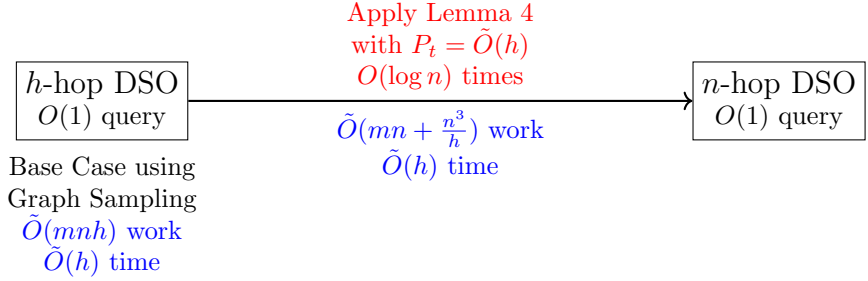
\begin{figure}
    \centering
    \begin{tikzpicture}
        
        \tikzset{shp/.style={color=black,draw}}
        \tikzset{every text node part/.style={align=center}}
        
        \node (h1) [shape=rectangle,shp] at (0,0) {{\large $h$-hop DSO} \\ $O(1)$ query};
        \node (h1txt) [below=0.1cm of h1,align=left] {Base Case using \\ Graph Sampling \\ \textcolor{blue}{$\tilde{O}(mnh)$ work} \\ \textcolor{blue}{$\tilde{O}(h)$ time}};\

        \node (h2) [shape=rectangle,shp] at (9,0) {{\large $n$-hop DSO} \\ $O(1)$ query};
        \path[draw, thick, ->] (h1) edge node [above] {\textcolor{red}{Apply Lemma~\ref{lem:pram:extend}} \\ \textcolor{red}{with $P_t=\tilde{O}(h)$} \\ \textcolor{red}{ $O(\log n)$ times}} node [below] {\textcolor{blue}{$\tilde{O}(mn + \frac{n^3}{h})$ work} \\ \textcolor{blue}{$\tilde{O}(h)$ time}} (h2) ;

        % \node (h4) [shape=rectangle,shp] at (9,-5) ;
        % \path[draw, thick, ->, dashed] (h2) edge node [left] {\textcolor{red}{Apply Lemma~\ref{lem:pram:extend}} \\ \textcolor{red}{ $O(\log n)$ times}} node [right] {\textcolor{blue}{$\tilde{O}(mn + \frac{n^3}{h})$ work} \\ \textcolor{blue}{$\tilde{O}(h)$ time}} (h4) ;
    \end{tikzpicture}
    \caption{Visual representation of preprocessing algorithm \textit{DSO-C}.}
    \label{fig:pram:algc}
\end{figure}

\textit{Base Case:}
% \vignesh{Changed reference to CONGEST DSO ARXIV paper, and included Lemma result.}
We present a method to construct an $h$-hop DSO using a graph sampling approach used in a distributed \congest DSO algorithm in~\cite{ManoharanR25dsoarxiv}, inspired by results in~\cite{WeimannY13}, captured by the following Lemma:

\begin{lemma}[Lemma 10, \cite{ManoharanR25dsoarxiv}]
    Sample $\tilde{h}=15 h \log n$ graphs $G_1,\cdots, G_{\tilde{h}}$ as subgraphs of $G$ by including each edge of $G$ with probability $1-\frac{1}{h}$. Then,
    \begin{enumerate}[label=\Alph*.]
        \item W.h.p. in $n$, there are $\Theta(\log n)$ graphs $G_i$ such that $e \not\in G_i$.
        \item Let $P_{sx}^e$ be a $h$-hop replacement path from $s$ to $x$ avoiding $e$. W.h.p. in $n$, there is at least one graph $G_i$ such that $G_i$ contains all edges in $P_{sx}^e$ but does not contain $e$.
    \end{enumerate}
    \label{lem:sampgraph}
\end{lemma}

We sample $\tilde{O}(h)$ graphs $G_i$ by removing every edge with probability $(1/h)$, independently at random. Then, any path $\mathcal{P}$ of $h$ hops not containing edge $e$, w.h.p. in $n$, there is at least one graph $G_i$  that includes all edges of $\mathcal{P}$ but not $e$. Additionally for a specific edge $e$, w.h.p. in $n$, there are only $O(\log n)$ graphs $G_i$ that do not contain $e$. If $d_i(x,y)$ is the shortest path distance in $G_i$, then $d(x,y,e) = \min_{i : e \not\in G_i} d_i(x,y)$, w.h.p. in $n$. 

To construct an $h$-hop DSO, we sample $\tilde{O}(h)$ graphs as described above and then compute $h$-hop limited APSP distances in each sampled graph $G_i$.  
We compute $h$-hop SSSP from each vertex of $G_i$ in $\tilde{O}(mn)$ work and $\tilde{O}(h)$ parallel time using an $\tilde{O}(m)$-work $\tilde{O}(h)$-time hop-limited shortest path algorithm (Lemma 5.3 of~\cite{CaoF23parallel}). 
We also store the identities of the set of graphs $G_i$ that do not contain $e$, for each edge $e \in E$. Thus, after preprocessing, we can answer a query in $O(\log n)$ time using the precomputed $d_i(x,y)$ distances. Since we have $\tilde{O}(h)$ graphs, we obtain a DSO with preprocessing cost of $\tilde{O}(mnh)$ work and $\tilde{O}(h)$ parallel time, with $\tilde{O}(1)$ work query. 

\textit{Obtaining complete DSO:}
Similar to \textsc{DSO-B}, we repeatedly apply Lemma~\ref{lem:pram:extend} $O(\log n)$ times to construct an $n$-hop DSO from an $h$-hop DSO. We use parameters $P_w = \tilde{O}(mnh)$ and $P_t = \tilde{O}(h)$. In an extension step where we construct a $\frac{3}{2}k$-hop DSO from a $k$-hop DSO for some $k \ge h$, we incur an additional preprocessing cost of $\tilde{O}(\frac{n^3}{k}+\frac{n^3}{k^3}) = \tilde{O}(\frac{n^3}{h})$ work and $\tilde{O}(h)$ parallel time. Adding up the costs for $O(\log n)$ steps, preprocessing cost is $\tilde{O}(mnh + \frac{n^3}{h})$ work and $\tilde{O}(h)$ parallel time. We only store preprocessed data for the final $n$-DSO, so the oracle has size $\tilde{O}(n^2)$.

\section{PRAM algorithms for other graph problems} 
\label{sec:pram:algother}

We present PRAM algorithms for other graph problems in the sequential $\tilde{O}(mn)$ time fine-grained complexity class: RPaths, 2-SiSP, 2-APSiSP, MWC and ANSC. We define these problems as follows:

% \vignesh{Definitions from thesis, modified for just PRAM.}

\begin{restatable}{definition}{defrpaths}
\textit{Replacement Paths (RPaths)} : Given a pair of vertices $s^*,t^* \in V$ and a shortest path $P_{s^*t^*}$ from $s^*$ to $t^*$, for each edge $e \in P_{s^*t^*}$, compute $d(s^*,t^*,e)$, the $s^*$-$t^*$ shortest path distance when $e$ is removed.
    \label{def:rpaths}
\end{restatable}

\begin{restatable}{definition}{defsisp}
\textit{Second Simple Shortest Path (2-SiSP)} : Given a pair of vertices $s^*,t^*$ and a shortest path $P_{s^*t^*}$ from $s^*$ to $t^*$, compute the weight $d_2(s^*,t^*)$ of a shortest simple path $\mathcal{P}_2$ from $s^*$ to $t^*$ that differs from $P_{s^*t^*}$ by at least one edge.
    \label{def:sisp}
\end{restatable}

\begin{definition}
    \textit{All Pairs Second Simple Shortest Path (2-APSiSP)}: 
    Compute $d_2(x,y)$, the $x$-$y$ second simple shortest path (2-SiSP) distance, for all pairs of vertices $x,y \in V$.
\end{definition}

\begin{restatable}{definition}{defmwc}
\textit{Minimum Weight Cycle problem (MWC)}: Compute the weight of a shortest simple cycle in $G$.
    \label{def:mwc}
\end{restatable}

\begin{restatable}{definition}{defansc}
\textit{All Nodes Shortest Cycles (ANSC)}: For each vertex $u \in V$, compute the weight of a shortest simple cycle through $u$ in $G$.
    \label{def:ansc}
\end{restatable}

Our main strategy in developing PRAM algorithms for these problems is to adapt existing sequential algorithms, leveraging work-optimal PRAM subroutines for APSP~\cite{KarczmarzS21} and SSSP~\cite{CaoF23parallel}. 

\subsection{RPaths and 2-SiSP}

In this section, we present our algorithms for RPaths and 2-SiSP in different graph types. Our algorithm for directed weighted RPaths and 2-SiSP uses a subroutine for APSP, adapting a sequential result from~\cite{AgarwalR18finegrained}. Our directed unweighted algorithm adapts a sequential result from~\cite{RodittyZ12} and uses SSSP as a subroutine. Finally, our undirected (weighted or unweighted) algorithm uses SSSP in addition to a characterization of replacement paths from~\cite{KatohIM82}. 

\begin{restatable}{theorem}{thmpramrpaths}
    Given a graph $G=(V,E)$, vertices $s^*, t^*$, we can compute RPaths and 2-SiSP in the PRAM model in the following bounds:
    \begin{enumerate}[label=\Alph*.,ref=\Alph*]
        \item \label{thm:pram:rpaths:dirwt} If $G$ is directed and weighted, the algorithm takes $\tilde{O}(mn)$ work and $\tilde{O}(n^{1/3})$ parallel time.
        \item \label{thm:pram:rpaths:dirunwt} If $G$ is directed and unweighted, the algorithm takes $\tilde{O}(m\sqrt{n})$ work and $\tilde{O}(\sqrt{n})$ parallel time.
        \item \label{thm:pram:rpaths:undir} If $G$ is undirected (weighted or unweighted), the algorithm takes $\tilde{O}(m)$ work and $\tilde{O}(\sqrt{n})$ parallel time.
    \end{enumerate}
\end{restatable}
\begin{proof}

    \begin{enumerate}[label=\Alph*.,ref=\Alph*]

    \item For directed weighted RPaths, we employ a reduction from RPaths to APSP in~\cite{AgarwalR18finegrained}. A similar reduction was explored for a distributed directed weighted RPaths algorithm in~\cite{ManoharanR24rpaths}; here, we adapt it to the PRAM model. This reduction involves constructing a modified graph, a step that requires only $O(m+n)$ work and $O(1)$ parallel time in the PRAM model. The dominant cost is the subsequent APSP computation. By applying the work-optimal sparse APSP algorithm of~\cite{KarczmarzS21} (which takes $\tilde{O}(mn)$ work and $\tilde{O}(n^{1/3})$ parallel time), we achieve the stated bounds. Computing 2-SiSP is straightforward once RPaths has been computed, as we just compute the minimum among replacement path distances in $O(n)$ work and $\tilde{O}(1)$ parallel time.

    \item For directed unweighted RPaths, we implement a sequential algorithm in~\cite{RodittyZ12} on the PRAM. This algorithm handles short and long paths separately. For short paths, it constructs $\sqrt{n}$ auxiliary graphs and computes SSSP in each. For long paths, it samples $\tilde{\Theta}(\sqrt{n})$ vertices and computes SSSP from each sampled source. In the PRAM model, these $\tilde{O}(\sqrt{n})$ SSSP computations can be executed in parallel. We can readily verify that the computation apart from the SSSP subroutines can be performed in $O(m+n)$ work and $\tilde{O}(1)$ parallel time.
    Using the $\tilde{O}(m)$ work and $\tilde{O}(\sqrt{n})$ time SSSP algorithm from~\cite{CaoF23parallel} for each of the $\tilde{O}(\sqrt{n})$ instances yields the total $\tilde{O}(m\sqrt{n})$ work and $\tilde{O}(\sqrt{n})$ parallel time. 
    
    Whether this parallel time can be improved while maintaining work-optimality remains an interesting open problem.
    
    \item For undirected RPaths, a characterization of replacement paths in~\cite{KatohIM82} gives us a method to compute RPaths using SSSP computations. Implementing this algorithm in the PRAM setting, using the work-optimal SSSP algorithm of~\cite{CaoF23parallel}, yields a PRAM algorithm with $\tilde{O}(m)$ work and $\tilde{O}(\sqrt{n})$ parallel time.
    
    \end{enumerate}
\end{proof}

\subsection{MWC and ANSC}  

We now present our PRAM algorithms for MWC and ANSC. For both directed graphs and undirected graphs, we use a reduction to APSP and then apply the PRAM APSP algorithm of~\cite{KarczmarzS21}. 

\begin{restatable}{theorem}{thmprammwc}
    Given a graph $G=(V,E)$, we can compute MWC and ANSC in the PRAM model in the $\tilde{O}(mn)$ work and $\tilde{O}(n^{1/3})$ parallel time.
\end{restatable}
\begin{proof}
    The reduction from directed ANSC to APSP requires computing $\min_{u \in V} d(u,v)+w(v,u)$ after all APSP distances $d(u,v)$ are computed. Clearly, the minimum computation can be implemented in PRAM in $\tilde{O}(n^2)$ work and $\tilde{O}(1)$ time. MWC can be readily computed after ANSC distances are computed.
    
    A similar reduction from undirected ANSC to APSP is more involved but again involves computing minimums among $O(n^2)$ values (see~\cite{AgarwalR18finegrained,WilliamsW18}).

    In both cases, the bottleneck is the initial APSP computation. Therefore, applying the APSP algorithm of~\cite{KarczmarzS21} directly yields the $\tilde{O}(mn)$ work and $\tilde{O}(n^{1/3})$ parallel time bounds.
\end{proof}

\subsection{2-APSiSP}
\label{sec:apsisp}

For 2-APSiSP, a work-optimal algorithm can be obtained by implementing a sequential algorithm of~\cite{AgarwalR16sisp}, using ideas from a distributed CONGEST algorithm for 2-APSiSP in~\cite{ManoharanR25dso}(TODO Alg refreence). However, its parallel time of $\tilde{O}(n)$ is notably slower than the $\tilde{O}(n^{1/3})$ and $\tilde{O}(\sqrt{n})$ bounds achieved for the other problems. We leave improving the parallel time for a work-optimal algorithm as an open problem.

Another approach to compute 2-APSiSP is to apply our DSO algorithms, and then use DSO queries to compute 2-SiSP distances. Using DSO-B gives us an $\tilde{O}(n^3)$ work and $\tilde{O}(1)$ parallel time algorithm, which is work-optimal for dense graphs.

\begin{restatable}{theorem}{thmpramapsisp}
    Given a directed weighted graph $G=(V,E)$, we can compute 2-APSiSP in the PRAM model in:
    \begin{enumerate}[label=\Alph*.,ref=\Alph*]
        \item $\tilde{O}(mn)$ work and $\tilde{O}(n)$ parallel time.
        \item $\tilde{O}(n^3)$ work and $\tilde{O}(1)$ parallel time.
    \end{enumerate}
\end{restatable}
\begin{proof}
    We use two different methods:

    % \vignesh{Added CONGEST algorithm to appendix, and refer to the lines here.}
    \noindent
    A. For our first algorithm, we build on the sequential algorithm of~\cite{AgarwalR16sisp}. For ease of explanation, we follow the steps of a distributed CONGEST implementation of this 2-APSiSP algorithm presented as Algorithm 6 in~\cite{ManoharanR25dsoarxiv}. For completeness, we repeat this distributed implementation in the appendix as Algorithm~\ref{alg:apsisp}.

    The APSP computation and the $n$ exclude computations used in Algorithm~\ref{alg:apsisp} can be implemented in a work-optimal manner in $\tilde{O}(mn)$ work and $\tilde{O}(n^{1/2})$ parallel time: APSP only needs $\tilde{O}(n^{1/3})$ parallel time, and implementing an exclude computation requires the same time as an SSSP computation~\cite{DemetrescuTCR08dso}. 
    
    Implementing the loop in Lines 7-13 of Algorithm~\ref{alg:apsisp} however, requires up to $n$ iterations of extracting elements from a priority queue to finalize each 2-SiSP distance. While the computation for paths ending at different vertices $y \in V$ (i.e., handling $H_y$) can be parallelized, the computation within $H_y$ seems to be sequential and requires $O(n)$ parallel time.

    The work done per vertex $y$ (for all $x \in V$) is still $\tilde{O}(n)$, and computing $d_2(x,y)$ for all $y \in V$ requires $\tilde{O}(n^2)$ work (subsumed by $\tilde{O}(mn)$). Thus, the total work is $\tilde{O}(mn)$, and the parallel time is $\tilde{O}(n)$.

    \noindent
    B. Our second approach is to use the PRAM DSO algorithms that we designed in this paper. Once a DSO is computed, the 2-SiSP distance can be computed as $d_2(x,y) = \min_{e \in P_{xy}} d(x,y,e)$, where $P_{xy}$ is an $x$-$y$ shortest path for $x,y \in V$. Thus, after computing APSP (which is done as part of DSO preprocessing) to get the edges in each path, we need $|P_{xy}| \le n$ queries to compute $d_2(x,y)$ for a single vertex pair $x,y$. Thus, 2-APSiSP requires $O(n^3)$ queries for all $O(n^2)$ vertex pairs.

    Thus, if we have a PRAM DSO algorithm with $O(1)$ work to answer any query, and preprocessing cost $P_w$ work and $P_t$ parallel time, then we obtain a 2-APSiSP algorithm with  $P_w+O(n^3)$ work and $P_t+\tilde{O}(1)$ parallel time. Using DSO-B with $P_w=\tilde{O}(n^3)$ and $P_t=\tilde{O}(1)$ proves our result. 
\end{proof}

\section{Conclusion and Open Problems} 
In this paper, we have presented the first non-trivial PRAM algorithms for DSO under single edge failure: We constructed a work-efficient algorithm for dense graphs with $\tilde{O}(n^3)$ work and $\tilde{O}(1)$ parallel time, matching APSP. We also presented a work-efficient algorithm for sparse graphs with preprocessing work $\tilde{O}(mn)$ matching sequential DSO~\cite{BernsteinK09dso} but with parallel time $\tilde{O}(n^{1/2+o(1)})$. While we have presented a work-time tradeoff algorithm that can achieve sub-$\sqrt{n}$ parallel time, it is not work-efficient. 

We have obtained work-efficient PRAM algorithms for many other graph problems in the sequential $\tilde{O}(mn)$ time fine-grained complexity class, including RPaths, 2-SiSP, MWC, ANSC, 2-APSiSP. For 2-APSiSP, we have presented two algorithms: a work-efficient $\tilde{O}(mn)$ algorithm with $\tilde{O}(n)$ parallel time, and an $\tilde{O}(n^3)$ work and $\tilde{O}(1)$ parallel time algorithm that is work-efficient for dense graphs.

Some open problems remain:

\begin{itemize}
    \item The main open problem for DSO is whether there is a work-efficient DSO preprocessing algorithm with $\tilde{O}(mn)$ work and $\tilde{O}(n^{1/3})$ parallel time, matching the current best work-efficient PRAM algorithm for APSP~\cite{KarczmarzS21}.
    \item A larger open problem is whether we can reduce the parallel time beyond $n^{1/3}$ for both APSP and DSO, while maintaining $\tilde{O}(mn)$ work. 
    \item For 2-APSiSP, can we improve the $\tilde{O}(n)$ parallel time for our work-efficient $\tilde{O}(mn)$ work algorithm?
\end{itemize}

\appendix

\section{2-APSiSP distributed implementation}

% \vignesh{Added appendix}
We include for reference a distributed algorithm for 2-APSiSP, referenced in the PRAM 2-APSiSP algorithm presented in this paper in Section~\ref{sec:apsisp}. The following Algorithm~\ref{alg:apsisp} appeared in~\cite{ManoharanR25dso} as Algorithm~6, and is a CONGEST implementation of a sequential algorithm in~\cite{AgarwalR16sisp}.

\begin{algorithm}[!h]
    \caption{\congest 2-APSiSP Algorithm (Algorithm~6 from~\cite{ManoharanR25dso})}
    \begin{algorithmic}[1]
        \Require Directed weighted graph $G=(V,E)$.
        \Ensure Compute 2-SiSP distance $d_2(x,y)$ at $y$ for each $x,y \in V$.
        \State Compute APSP in $G$ and $G^r$, remembering parent nodes. \label{alg:capsisp:apsp}
        \State  For each vertex $y \in V$, initialize $H_y = \emptyset$. $H_y$ is a priority queue for Type-B paths \textit{ending} at $y$. $H_y$ will contain keys $x \in V$ with priority $D$ representing $x$-$y$ 2-SiSP of weight $D$. \rightComment{$H_y$ is used in line~\ref{alg:capsisp:type2f}, lines~\ref{alg:capsisp:secst}-\ref{alg:capsisp:secen}.} 
        \ForAll{vertex $x \in V$} \label{alg:capsisp:initst}
            \State Find all neighbors $a$ such that edge $(x,a)$ is a shortest path. 
            \State Perform exclude computation with source $x$, and edges $(x,a)$ as the excluded set of independent paths (see Section 4 of~\cite{ManoharanR25dsoarxiv}): Distance $d(x,y,(x,a))$ is computed at $y$. \label{alg:capsisp:exclude} \rightComment{Used in lines~\ref{alg:capsisp:excludeset},\ref{alg:capsisp:type2f}}
            \State  Perform a downcast on $T_x$ (out-shortest path tree for source $x$), propagating $(x,a)$ and $w(x,a)$ to each vertex $y$ in the subtree rooted at $a$. \rightComment{Propagate first-edge info, used in lines~\ref{alg:capsisp:extendcreate},\ref{alg:capsisp:type2f}}  \label{alg:capsisp:downcast}
            \For{vertex $a$} 
            \For{vertex $y \in V$ in subtree of $T_x$ rooted at $a$} 
                \lineComment{Lines~\ref{alg:capsisp:excludeset}-\ref{alg:capsisp:initen}: Local computation at $y$. Sets Type-A distance in line~\ref{alg:capsisp:excludeset} and initializes Type-B distance computation.}
                \State Let $d_2^*(x,y) \gets d(x,y,(x,a))$ be the current estimate of 2-SiSP distance. \rightComment{Set Type-A distance, computed in line~\ref{alg:capsisp:exclude}} \label{alg:capsisp:excludeset}
                \State Add $(x,a)$ to $Extensions(a,y)$. \rightComment{Used in line~\ref{alg:capsisp:extend} for Type-B paths } \label{alg:capsisp:extendcreate}
                \lineComment{Add Type-B candidate for $x$-$y$ 2-SiSP to $H_y$.}
                \State If $(a,b)$ is the first edge of $a$-$y$ shortest path and $b \ne y$, add key $x$ to $H_y$ with weight $w(x,a) + d(a,y,(a,b))$. \rightComment{Note that $(a,b)$ has been propagated from the downcast with source $a$ from line~\ref{alg:capsisp:downcast}, and $d(a,y,(a,b))$ is computed at $y$ in line~\ref{alg:capsisp:exclude}.} \label{alg:capsisp:initen} \label{alg:capsisp:type2f}
            \EndFor
            \EndFor
        \EndFor 
        \lineComment{Lines \ref{alg:capsisp:secst}-\ref{alg:capsisp:secen}: Compute shortest Type-B 2-SiSP distance}
        \For{vertex $y \in V$} \label{alg:capsisp:secst}
        \lineComment{Local computation at $y$ to compute $d_2(x,y)$ for all $x \in V$.}
            \While{$H_y \ne \emptyset$}
            \State Let $(x,D) \gets \textsc{Extract-min}(H_y)$ \label{alg:capsisp:extract}
            \If{$D \ge d_2^*(x,y)$}
                \State Set $d_2(x,y) \gets d_2^*(x,y)$ \rightComment{Type-A path is shortest.} \label{alg:capsisp:found1}
            \Else
                \State Set $d_2(x,y) \gets D$ \rightComment{Type-B path is shortest.} \label{alg:capsisp:found}
                \For{$(x',x) \in Extensions(x,y)$}
                    \State Add key $x'$ with weight $w(x',x)+d_2(x,y)$ to $H_y$. \rightComment{Update Type-B candidate path distances for $x'$-$y$ 2-SiSP.} \label{alg:capsisp:extend} \label{alg:capsisp:secen}
                \EndFor
            \EndIf
            \EndWhile
        \EndFor
    \end{algorithmic}
    \label{alg:apsisp}
\end{algorithm}

\bibliographystyle{plainurl}
\bibliography{references}

@inproceedings{ManoharanR25dso,
  author       = {Vignesh Manoharan and
                  Vijaya Ramachandran},
  title        = {Distributed Distance Sensitivity Oracles},
  booktitle    = {Structural Information and Communication Complexity - 32nd International
                  Colloquium, {SIROCCO} 2025, Delphi, Greece, June 2-4, 2025, Proceedings},
  series       = {Lecture Notes in Computer Science},
  volume       = {15671},
  pages        = {366--383},
  publisher    = {Springer},
  year         = {2025},
  note = {Invited to the Special Issue of the journal Theoretical Computer Science devoted to
selected papers of SIROCCO’25}
}

@inproceedings{DeyG24,
  author       = {Dipan Dey and
                  Manoj Gupta},
  title        = {Nearly Optimal Fault Tolerant Distance Oracle},
  booktitle    = {Proceedings of the 56th Annual {ACM} Symposium on Theory of Computing,
                  {STOC} 2024, Vancouver, BC, Canada, June 24-28, 2024},
  pages        = {944--955},
  publisher    = {{ACM}},
  year         = {2024}
}

@article{ShiloachV81,
  author       = {Yossi Shiloach and
                  Uzi Vishkin},
  title        = {Finding the Maximum, Merging, and Sorting in a Parallel Computation
                  Model},
  journal      = {J. Algorithms},
  volume       = {2},
  number       = {1},
  pages        = {88--102},
  year         = {1981}
}

@inproceedings{CaoF23parallel,
  author       = {Nairen Cao and
                  Jeremy T. Fineman},
  title        = {Parallel Exact Shortest Paths in Almost Linear Work and Square Root
                  Depth},
  booktitle    = {Proceedings of the 2023 {ACM-SIAM} Symposium on Discrete Algorithms,
                  {SODA} 2023, Florence, Italy, January 22-25, 2023},
  pages        = {4354--4372},
  publisher    = {{SIAM}},
  year         = {2023}
}

@article{KatohIM82,
  author       = {Naoki Katoh and
                  Toshihide Ibaraki and
                  Hisashi Mine},
  title        = {An efficient algorithm for {K} shortest simple paths},
  journal      = {Networks},
  volume       = {12},
  number       = {4},
  pages        = {411--427},
  year         = {1982}
}

@inproceedings{BernsteinK08dso,
  author       = {Aaron Bernstein and
                  David R. Karger},
  title        = {Improved distance sensitivity oracles via random sampling},
  booktitle    = {Proceedings of the Nineteenth Annual {ACM-SIAM} Symposium on Discrete
                  Algorithms, {SODA} 2008, San Francisco, California, USA, January 20-22,
                  2008},
  pages        = {34--43},
  publisher    = {{SIAM}},
  year         = {2008}
}

@article{Yen1971,
  title={Finding the k shortest loopless paths in a network},
  author={Yen, Jin Y},
  journal={Management Science},
  volume={17},
  number={11},
  pages={712--716},
  year={1971},
  publisher={Informs}
}

@inproceedings{AgarwalR16sisp,
  author       = {Udit Agarwal and
                  Vijaya Ramachandran},
  title        = {Finding k Simple Shortest Paths and Cycles},
  booktitle    = {27th International Symposium on Algorithms and Computation, {ISAAC}
                  2016, December 12-14, 2016, Sydney, Australia},
  series       = {LIPIcs},
  volume       = {64},
  pages        = {8:1--8:12},
  publisher    = {Schloss Dagstuhl - Leibniz-Zentrum f{\"{u}}r Informatik},
  year         = {2016}
}

@inproceedings{AgarwalR18finegrained,
  author       = {Udit Agarwal and
                  Vijaya Ramachandran},
  title        = {Fine-grained complexity for sparse graphs},
  booktitle    = {Proceedings of the 50th Annual {ACM} {SIGACT} Symposium on Theory
                  of Computing, {STOC} 2018, Los Angeles, CA, USA, June 25-29, 2018},
  pages        = {239--252},
  publisher    = {{ACM}},
  year         = {2018}
}

@article{RodittyZ12,
  author       = {Liam Roditty and
                  Uri Zwick},
  title        = {Replacement paths and \emph{k} simple shortest paths in unweighted
                  directed graphs},
  journal      = {{ACM} Trans. Algorithms},
  volume       = {8},
  number       = {4},
  pages        = {33:1--33:11},
  year         = {2012}
}

@misc{ManoharanR25dsoarxiv,
      title={Distributed Distance Sensitivity Oracles}, 
      author={Vignesh Manoharan and Vijaya Ramachandran},
      year={2025},
      eprint={2411.13728},
      archivePrefix={arXiv},
      primaryClass={cs.DS},
      url={https://arxiv.org/abs/2411.13728}, 
}

@article{DemetrescuTCR08dso,
  author       = {Camil Demetrescu and
                  Mikkel Thorup and
                  Rezaul Alam Chowdhury and
                  Vijaya Ramachandran},
  title        = {Oracles for Distances Avoiding a Failed Node or Link},
  journal      = {{SIAM} J. Comput.},
  volume       = {37},
  number       = {5},
  pages        = {1299--1318},
  year         = {2008}
}

@inproceedings{BernsteinK09dso,
  author       = {Aaron Bernstein and
                  David R. Karger},
  title        = {A nearly optimal oracle for avoiding failed vertices and edges},
  booktitle    = {Proceedings of the 41st Annual {ACM} Symposium on Theory of Computing,
                  {STOC} 2009, Bethesda, MD, USA, May 31 - June 2, 2009},
  pages        = {101--110},
  publisher    = {{ACM}},
  year         = {2009}
}

@article{WeimannY13,
  author       = {Oren Weimann and
                  Raphael Yuster},
  title        = {Replacement Paths and Distance Sensitivity Oracles via Fast Matrix
                  Multiplication},
  journal      = {{ACM} Trans. Algorithms},
  volume       = {9},
  number       = {2},
  pages        = {14:1--14:13},
  year         = {2013}
}

@article{Ren22,
  author       = {Hanlin Ren},
  title        = {Improved distance sensitivity oracles with subcubic preprocessing
                  time},
  journal      = {J. Comput. Syst. Sci.},
  volume       = {123},
  pages        = {159--170},
  year         = {2022}
}

@article{WilliamsW18,
  author       = {Virginia {Vassilevska Williams} and
                  R. Ryan Williams},
  title        = {Subcubic Equivalences Between Path, Matrix, and Triangle Problems},
  journal      = {J. {ACM}},
  volume       = {65},
  number       = {5},
  pages        = {27:1--27:38},
  year         = {2018}
}

@inproceedings{ManoharanR24rpaths,
  author       = {Vignesh Manoharan and
                  Vijaya Ramachandran},
  title        = {Computing Replacement Paths in the {CONGEST} Model},
  booktitle    = {Structural Information and Communication Complexity - 31st International
                  Colloquium, {SIROCCO} 2024, Vietri sul Mare, Italy, May 27-29, 2024,
                  Proceedings},
  series       = {Lecture Notes in Computer Science},
  volume       = {14662},
  pages        = {420--437},
  publisher    = {Springer},
  year         = {2024}
}

@article{KleinS97,
  author       = {Philip N. Klein and
                  Sairam Subramanian},
  title        = {A Randomized Parallel Algorithm for Single-Source Shortest Paths},
  journal      = {J. Algorithms},
  volume       = {25},
  number       = {2},
  pages        = {205--220},
  year         = {1997}
}

@article{Spencer97,
  author       = {Thomas H. Spencer},
  title        = {Time-work tradeoffs for parallel algorithms},
  journal      = {J. {ACM}},
  volume       = {44},
  number       = {5},
  pages        = {742--778},
  year         = {1997}
}

@inproceedings{RozhonGHZL22undir,
  author       = {V{\'{a}}clav Rozhon and
                  Christoph Grunau and
                  Bernhard Haeupler and
                  Goran Zuzic and
                  Jason Li},
  title        = {Undirected (1+\emph{{\(\epsilon\)}})-shortest paths via minor-aggregates:
                  near-optimal deterministic parallel and distributed algorithms},
  booktitle    = {{STOC} '22: 54th Annual {ACM} {SIGACT} Symposium on Theory of Computing,
                  Rome, Italy, June 20 - 24, 2022},
  pages        = {478--487},
  publisher    = {{ACM}},
  year         = {2022}
}

@article{HanPR97,
  author       = {Yijie Han and
                  Victor Y. Pan and
                  John H. Reif},
  title        = {Efficient Parallel Algorithms for Computing All Pair Shortest Paths
                  in Directed Graphs},
  journal      = {Algorithmica},
  volume       = {17},
  number       = {4},
  pages        = {399--415},
  year         = {1997}
}

@article{UllmanY91,
  author       = {Jeffrey D. Ullman and
                  Mihalis Yannakakis},
  title        = {High-Probability Parallel Transitive-Closure Algorithms},
  journal      = {{SIAM} J. Comput.},
  volume       = {20},
  number       = {1},
  pages        = {100--125},
  year         = {1991}
}

@article{BerkmanV93,
  author       = {Omer Berkman and
                  Uzi Vishkin},
  title        = {Recursive Star-Tree Parallel Data Structure},
  journal      = {{SIAM} J. Comput.},
  volume       = {22},
  number       = {2},
  pages        = {221--242},
  year         = {1993}
}

@inproceedings{KarczmarzS21,
  author       = {Adam Karczmarz and
                  Piotr Sankowski},
  title        = {A Deterministic Parallel {APSP} Algorithm and its Applications},
  booktitle    = {Proceedings of the 2021 {ACM-SIAM} Symposium on Discrete Algorithms,
                  {SODA} 2021, Virtual Conference, January 10 - 13, 2021},
  pages        = {255--272},
  publisher    = {{SIAM}},
  year         = {2021}
}

\end{document}